\documentclass[conference]{IEEEtran}
\IEEEoverridecommandlockouts

\usepackage{cite}
\usepackage{amsmath,amssymb,amsfonts}
\usepackage{algorithmic}
\usepackage{graphicx}
\usepackage{textcomp}
\usepackage{xcolor}

\usepackage{braket}
\usepackage{subcaption}
\usepackage{amsthm, tikz}
\usepackage{float}

\newtheorem{theorem}{Theorem}

\newtheorem{corollary}{Corollary}

\def\BibTeX{{\rm B\kern-.05em{\sc i\kern-.025em b}\kern-.08em
    T\kern-.1667em\lower.7ex\hbox{E}\kern-.125emX}}
\begin{document}

\title{Phantom Edges in the Problem Hamiltonian: A Method for Increasing Performance and Graph Visibility for QAOA}

\author{\IEEEauthorblockN{Quinn Langfitt}
\IEEEauthorblockA{\textit{Theoretical Division} \\
\textit{Los Alamos National Laboratory}\\
Los Alamos, NM 87544 \\
Email: qlangfitt@u.northwestern.edu}
\and
\IEEEauthorblockN{Reuben Tate}
\IEEEauthorblockA{\textit{CCS-3: Information Sciences} \\
\textit{Los Alamos National Laboratory}\\
Los Alamos, NM 87544 \\
Email: rtate@lanl.gov}
\and
\IEEEauthorblockN{Stephan Eidenbenz}
\IEEEauthorblockA{\textit{CCS-3: Information Sciences} \\
\textit{Los Alamos National Laboratory}\\
Los Alamos, NM 87544 \\
}
}

\maketitle


\begin{abstract}
The Quantum Approximate Optimization Algorithm (QAOA) is a variational quantum algorithm that can be used to approximately solve combinatorial optimization problems. However, a major limitation of QAOA is that it is a ``local" algorithm for finite circuit depths, meaning it can only optimize over local properties of the graph. In this paper, we present Phantom-QAOA, a new QAOA ansatz that introduces only one additional parameter to the standard ansatz---regardless of system size---allowing QAOA to ``see'' more of the graph at a given depth $p$. We achieve this by modifying the target graph to include additional $\alpha$-weighted edges, with $\alpha$ serving as a tunable parameter. This modified graph is then used to construct the phase operator and allows QAOA to explore a wider range of the graph's features. We derive a general formula for our new ansatz at $p=1$ and analytically show an improvement in the approximation ratio for cycle graphs. We also provide numerical experiments that demonstrate significant improvements in the approximation ratio for the Max-Cut problem over the standard QAOA ansatz for $p=1$ and $p=2$ on random regular graphs up to 16 nodes.
\end{abstract}


\section{Introduction}
\label{sec:introduction}
Over the past several years, researchers have studied the potential ability of quantum computers  to outperform classical computers on a variety of problems and tasks. One such category of problems are combinatorial optimization problems, where -- given a finite discrete set of candidate solutions -- the objective is to find an optimal solution with respect to some cost function. Such problems appear in many different domains including power systems \cite{soares2018survey}, logistics \cite{BLN18}, VLSI \cite{barahona1988application}, and more. Our work focuses on the well-studied Max-Cut problem, which is to determine an optimal partition of the vertices of the graph into two groups such that the number of edges between the two groups is maximized; however, the main ideas and techniques in our work can also be used for other combinatorial optimization problems.

For approximately a decade, the primary framework for solving combinatorial optimization problems on quantum devices has been the Quantum Approximate Optimization Algorithm (QAOA) introduced by Farhi et al. in 2014 \cite{farhi2014quantumapproximateoptimizationalgorithm}. Since its introduction, several variants and modifications of the QAOA algorithm have been proposed; we refer the interested reader to \cite{blekos2024review,keller2024quantum}, which contain overviews of such modifications. Modifications of the QAOA algorithm primarily fall into one of four types: (1) modification of the initial quantum state, (2) modification of the mixing Hamiltonian (which drives the evolution of the quantum state), (3) various schemes for initialization and optimization of the parameters of the quantum circuit, and (4) modification of the cost Hamiltonian, which encodes the cost function we wish to optimize. This work focuses on the last category (cost Hamiltonian modification) which is one of the least-studied types of QAOA modifications.

Since the cost Hamiltonian typically encodes the cost function of the underlying combinatorial optimization problem, it is not clear a priori what alternatives one should consider. In particular, with the exception of an affine shift in the cost function, any other cost Hamiltonian will most likely break QAOA's connection to Quantum Adiabatic Computing (QAC); this connection ensures that QAOA (with optimal parameters) approaches an optimal solution as the number of circuit layers $p$ tends to infinity. In this work, we consider an alternative cost Hamiltonian that is formed by a modified graph, which is constructed by starting with the original input graph and adding additional $\alpha$-weighted edges, with $\alpha$ being a tunable parameter. If one allows optimization of the $\alpha$-parameter, then this QAOA variant does at least as well as standard QAOA, as the standard QAOA is recovered at $\alpha = 0$. We call our new QAOA variant Phantom-QAOA.

Previous work by \cite{wilkie2024qaoarandomsubgraphphase} has considered QAOA with a cost Hamiltonian obtained by a modified graph; however, they primarily consider modified graphs that are a \emph{subgraph} of the original graph. The initial motivation for our approach is somewhat different: it is well-known that QAOA is a ``local" algorithm and thus ``seeing the whole graph" is a necessary condition to ensure that QAOA does well on arbitrary graphs; by adding additional edges to the cost Hamiltonian, this potentially allows QAOA to ``see the graph more quickly" as the depth increases. We discuss this motivation in more detail in the discussion section (Section \ref{sec:discussion}).

Similar to Multi-Angle QAOA---and a recent generalization, eXpressive QAOA \cite{Vijendran_2024}, which adds an extra Pauli-Y component in the mixer Hamiltonian---our method introduces additional parameters; however, Phantom-QAOA only adds a single parameter, regardless of problem size or number of QAOA layers. On the other hand, Multi-Angle QAOA adds an additional $p(n+m-2)$ parameters where $m$ is the number of terms in the original cost Hamiltonian; the addition of such parameters significantly increases the parameter search space, thus making the cost landscape increasingly difficult and computationally intensive to search over.

In this work, we demonstrate theoretical results for single-layer QAOA with our approach. For any way of adding $\alpha$-weighted edges, we provide a general formula for the expectation of an edge being cut for single-layer QAOA. Using this formula, we demonstrate an example where our approach does strictly better than standard QAOA; more specifically, for even-length cycles, we show that, when placing the additional edges in $G$ to form triangles, our ansatz yields an approximation ratio of $0.7925$ whereas standard QAOA only has an approximation ratio of $0.75$.

Additionally, we provide empirical results that show that when the parameter search space is restricted, certain additions of $\alpha$-weighted edges in the cost Hamiltonian yield a noticeable improvement in approximation ratio on random regular graphs. When the parameter search space is unrestricted, empirical results suggest that, regardless of whether or not the graph is regular, improvements can be obtained for any $0 < \alpha  < 1$ and that all such non-trivial $\alpha$ yield roughly the same level of improvement; the caveat is that such unrestricted search spaces are difficult to optimize over, e.g., optimal parameters may be significantly far from the origin.

\section{Background}
\label{sec:background}

\subsection{Max-Cut Problem}
The problem we focus on in this paper is the Max-Cut problem—a canonical NP-hard combinatorial optimization problem \cite{WOEGINGER2005210}. The problem is defined as follows: given an undirected graph \(G=(V,E)\), partition the vertex set \(V\) into two disjoint subsets \(S\) and \(V \setminus S\) in such a way that the number of edges with one endpoint in \(S\) and the other in \(V \setminus S\) is maximized, i.e.,
\[
\max_{S \subset V} \; \left|\{ (u,v) \in E : u \in S,\; v \in V \setminus S \}\right|.
\]
The best known classical approach to approximating this optimum is the Goemans-Williamson algorithm \cite{Goemans}, which employs semidefinite programming coupled with randomized rounding to achieve a worst-case approximation ratio of 0.878. 

An equivalent formulation assigns to each vertex \(i \in V\) a binary variable \(z_i \in \{0,1\}\), so that a candidate solution is represented by the vector \(z = (z_1, z_2, \dots, z_n)\). The corresponding cut value is given by
\[
C(z) = \frac{1}{2} \sum_{(u,v) \in E} \Bigl(1 - (-1)^{z_u+z_v}\Bigr).
\]
Alternatively, we can reformulate the above cost function in terms of computational basis states by using the standard Pauli-\(Z\) operator representation (with \(Z\ket{0} = \ket{0}\) and \(Z\ket{1} = -\ket{1}\)). The cost function can then be expressed as
\[
C = \frac{1}{2}\sum_{(u,v) \in E} \Bigl(1 - Z_u Z_v\Bigr).
\]

\subsection{Quantum Approximate Optimization Algorithm}

QAOA is a hybrid quantum-classical algorithm in which a parameterized quantum circuit generates sample solutions, and a classical optimizer varies the circuit parameters in order to maximize (or minimize) the objective function. The QAOA ansatz is inspired by adiabatic quantum computing \cite{Albash_2018}. The mixer operator $U(B, \beta_i)$ and the phase operator $U(C, \gamma_i)$ are applied in an alternating sequence, and the $\beta_i$ and $\gamma_i$ are the respective parameters for the $i$th layer. The number of layers, $p$, is called the circuit depth. 


The mixer operator is generated by the transverse-field Hamiltonian, which promotes transitions between different computational basis states. It is given by:
\[
U(B, \beta_i) = e^{-i \beta_i B}, \quad B = \sum_{j=1}^n X_j,
\]
where $X_j$ is the Pauli-X operator acting on the $j$th qubit.

The phase operator typically encodes the problem we wish to solve. For the Max-Cut problem, the phase operator is generated by the cost Hamiltonian $C$ that represents the cut value of a given partition of the graph's vertices. In our formulation, each vertex \(i\) in the graph \(G = (V,E)\) is associated with a binary variable \(z_i \in \{0,1\}\) so that a computational basis state \(|z_1z_2\cdots z_n\rangle\) corresponds to a particular partition of the vertices---for example, with \(z_i = 0\) indicating that vertex \(i\) belongs to one subset and \(z_i = 1\) indicating membership in the other. The phase operator is then given by:
\[
U(C, \gamma_i) = e^{-i \gamma_i C}, \quad C = \frac{1}{2} \sum_{(u,v) \in E} (1 - Z_u Z_v),
\]
where $Z_u$ and $Z_v$ are Pauli-Z operators acting on the qubits representing vertices $u$ and $v$, and $E$ is the set of edges in the graph $G = (V, E)$. 

The variational quantum state after $p$ layers is 
\[
|\psi(\boldsymbol{\gamma}, \boldsymbol{\beta})\rangle = \prod_{i=1}^p U(B, \beta_i) U(C, \gamma_i) |s\rangle,
\]

where $\ket{s}$ is the initial state. In this work, we follow the standard practice of initializing the state in the most excited state of the mixer Hamiltonian: $\ket{+}^{\otimes n}$. Note that this initial state is the equal superposition of all $2^n$ possible cuts in the graph. Equivalently, measuring $\ket{s}$ would cause each vertex to independently have equal probability of being on either side of the cut; one can show that this implies that each edge is cut with probability $1/2$ and hence measuring $\ket{s}$ would yield a cut value of $1/2$ the number of edges. 

To obtain the expected Max-Cut value, we conjugate the state with the Cost Hamiltonian: 

\begin{align}
    \langle C \rangle = \langle \psi(\boldsymbol{\gamma}, \boldsymbol{\beta})| C | \psi(\boldsymbol{\gamma}, \boldsymbol{\beta}) \rangle.
\end{align}

Maximizing the expectation value of $\langle C \rangle$ with respect to the parameters $\boldsymbol{\gamma}$, $\boldsymbol{\beta}$ is the objective of QAOA. Let $F_p$ represent the expectation value of the cost Hamiltonian after $p$ layers. The maximum achievable expectation value, $\text{max}_{\boldsymbol{\gamma}, \boldsymbol{\beta}} F_p(\gamma, \beta)$, increases with $p$. In the limit as $p \to \infty$, QAOA (with optimal parameters) is guaranteed to achieve an approximation ratio of $1$ \cite{farhi2014quantumapproximateoptimizationalgorithm}.



\begin{figure}[t]
    \centering
    \scalebox{0.75}{\begin{tikzpicture}[scale=1.5]
    \node[circle, draw, minimum size=2em, inner sep=2pt, label=center:$u$] (u) at (0,0) {};
    \node[circle, draw, minimum size=2em, inner sep=2pt, label=center:$v$] (v) at (2,0) {};
    \node[circle, draw, minimum size=2em, inner sep=2pt, label=center:$w$] (w) at (1,1.5) {};
    \node[circle, draw, minimum size=2em, inner sep=2pt, label=center:$x$, fill=green, fill opacity=0.5, text opacity=1] (x) at (0,-1.5) {};
    \node[circle, draw, minimum size=2em, inner sep=2pt, label=center:$y$, fill=green, fill opacity=0.5, text opacity=1] (y) at (2,-1.5) {};

    \draw[line width=0.8mm] (u) -- (v);
    \draw[line width=0.8mm] (u) -- (w);
    \draw[line width=0.8mm] (v) -- (w);
    \draw[line width=0.8mm] (u) -- (x);
    \draw[line width=0.8mm] (x) -- (y);
    \draw[red, line width=0.5mm] (v) -- (y);  // Red line from previous modification, now thinner.

    \draw[line width=0.8mm] (x) -- (w);  // Black line from x to w.
    \draw[line width=0.8mm] (y) -- (u);  // Black line from y to u.

    \draw[solid, red, line width=0.5mm] (w) -- (y);  // This connection remains red, but now thinner.
    \draw[solid, red, line width=0.5mm] (v) -- (x);  // This connection remains red, but now thinner.
\end{tikzpicture}}
    \caption{Visualization of graph $G'$, where the black lines are the edges from the original graph $G$ and additional $\alpha$-weighted edges are denoted by red lines. For the edge \((xy)\), \(d = 2\) and \(d' = 1\) denote two existing and one phantom neighbor of \(x\), respectively. Similarly, for vertex $y$, \(e = 1\) and \(e' = 2\). The triangles formed around this edge include: one with vertex \(u\) (\(f = 1\)), one with vertex \(w\) (\(f' = 1\)), and another with vertex \(v\) (\(f'' = 1\)).}
    \label{fig:example_graph}
\end{figure}

\subsection{Graph Visibility in QAOA}
One key limitation of QAOA is that it is a local algorithm for finite $p$. As observed by Farhi et al. \cite{farhi2014quantumapproximateoptimizationalgorithm, farhi2020typcalcase, farhi2020worstcase}, the total unitary applied to any edge term only depends on the structure of the graph within a radius of $p$ hops from that edge. The algorithm therefore optimizes over a linear combination of subgraphs of the target graph, where the subgraphs are weighted by the number of occurrences. As $p$ increases, the subgraph size increases until it covers the entire graph. 

To demonstrate QAOA's locality, let's evaluate a single edge term $C_{uv} = \frac{1}{2}(I - Z_u Z_v)$ in the cost Hamiltonian (recall that $C = \sum_{uv} C_{uv}$). The expected cut value of the edge $(u, v)$ at depth $p=1$ is

\[
\langle C_{uv} \rangle = \langle s| e^{i \gamma C} e^{i \beta B} C_{uv} e^{-i \beta B} e^{-i \gamma C} | s \rangle.
\]

However, adjusting the cost Hamiltonian that is used in the phase operatory $U(C, \gamma)$ can lead to changes in how much the QAOA can "see" for a given $p$. One example of the application of a modified phase operator is shown in \cite{wilkie2024qaoarandomsubgraphphase}. We discuss our new phase operator in detail in the next section. 

\section{Phase Operator with Phantom Edges}
\label{sec:modified_phase_operator}

Mathematically, we introduce a new phase operator $U(\tilde{C}, \gamma)$, where $\tilde{C}$ is the Hamiltonian including additional small-weight "phantom" edges:
\[ 
\tilde{C} = C + \alpha C' 
\]
Here, $C$ is the original Max-Cut cost Hamiltonian, $C'= \frac{1}{2}\sum_{(u,v) \in E'} \Bigl(1 - Z_u Z_v\Bigr)$ represents the cut value for the additional edges, and $\alpha$ is the weight of the additional edges. The modified phase operator is then:
\[ 
U(\tilde{C}, \gamma) = e^{-i \gamma \tilde{C}} = e^{-i \gamma (C + \alpha C')}.
\]

The optimization target remains the original cost Hamiltonian $C$ for the Max-Cut problem. The expectation value of $C$ with respect to the state $|\psi(\gamma, \beta)\rangle$ is:
\[ 
\langle C \rangle = \langle \psi(\boldsymbol{\gamma}, \boldsymbol{\beta}) | C | \psi(\boldsymbol{\gamma}, \boldsymbol{\beta}) \rangle 
\]
By incorporating additional edge terms in the phase operator, and allowing its weight $\alpha$ to be a tunable parameter, the ansatz should be able to explore a larger portion of the graph, increasing the likelihood of finding a better cut. This expectation value, with the modified phase operator and for $p=1$, can be expressed as:
\[ 
\langle C \rangle = \langle s| U^\dagger(\tilde{C}, \gamma) U^\dagger(B, \beta) C U(B, \beta) U(\tilde{C}, \gamma) |s \rangle.
\]

\subsection{Methods for Additional Weighted-Edge Placement}
\label{sec:constructionOfModifiedGraph}
In this paper, we present two methods for constructing the augmented graph $G'$, which is used to define the modified cost Hamiltonian $\tilde{C}$ that generates the phase operator $U(\tilde{C}, \gamma)$ in our QAOA ansatz. These approaches are referred to as the "Full" and "Triangle" methods. 

The Full method consists of adding $\alpha$-weighted edges between all pairs of vertices in $G$ that do not already share an edge, forming a fully connected graph $G'$. In other words, the set of edges that are added are simply the edges of the complement of $G$. Mathematically, if $E$ is the set of edges in the original graph $G$, the new set of edges in $G'$, denoted by $E'$, includes all possible edges between vertices: 
\[
E' = E \cup \{(u, v) \mid (u, v) \notin E, \, u, v \in V(G)\}.
\]
The additional edges in $G'$ are then weighted by the factor $\alpha$ in the new cost Hamiltonian $C'$.

The Triangle method adds $\alpha$-weighted edges only between vertices that are two hops away in $G$, forming triangles with the existing edges. Mathematically, for any vertex pair $(u, v)$ such that the shortest path between $u$ and $v$ in $G$ is exactly two edges, we add an edge $(u, v)$ to $G'$:
\[
E' = E \cup \{(u, v) \mid \text{dist}_G(u, v) = 2\}.
\]

The motivation behind this approach is that it has been observed that worst-case approximation ratios occur when the graph is triangle-free \cite{farhi2020worstcase}. 

\section{Derivation of Max-Cut Expectation Value for Phantom-QAOA}
\label{sec:Derivation}


The formula for the unweighted Max-Cut problem for the original QAOA ansatz at $p=1$ (QAOA$_1$) is contained in the following Theorem \cite{hadfield2018quantumalgorithmsscientificcomputing}:

\begin{theorem}
\label{thm:original_maxcut}
Consider the QAOA$_1$ state $|\gamma, \beta\rangle$ for Max-Cut on a graph $G$.
\begin{itemize}
\label{thm:general_unweighted_formula}
    \item For each edge $(uv)$,
    \begin{align*}
    \langle \gamma, \beta | C_{uv} | \gamma, \beta \rangle = \frac{1}{2} + \frac{1}{4} \sin(4\beta) \sin \gamma (\cos^d \gamma + \cos^e \gamma) \\ - \frac{1}{4} \sin^2(2\beta) \cos^{d+e-2f} \gamma (1 - \cos^f (2\gamma)) =: \langle C_{uv} \rangle(d, e, f)
    \end{align*}
    where $d = \deg(u) - 1$, $e = \deg(v) - 1$, and $f$ is the number of triangles in the graph containing $(uv)$.
    \item The overall expectation value is
    \[
    \langle C \rangle = \langle \gamma, \beta | C | \gamma, \beta \rangle = \sum_{(d,e,f)} \langle C_{uv} \rangle(d, e, f) \chi(d, e, f),
    \]
    where $\chi(d, e, f)$ gives the number of edges in $G$ with neighborhood parameters $(d, e, f)$.
\end{itemize}
\end{theorem}

In this section, we generalize Theorem~\ref{thm:original_maxcut} for Phantom-QAOA at $p=1$.
We acknowledge that the expression itself may not provide immediate insights or be easily interpretable. We provide an example in the following subsection (Section ~\ref{sec:analytic_example}) to improve interpretability as well as show the utility of introducing additional edges so that additional triangles are formed with the edges in $G$. For additional clarity on the parameters of the following theorem, we include an example graph with labeled parameters in Fig. ~\ref{fig:example_graph}.

\begin{theorem}
\label{thm:general_modified_ansatz_formula}
For each vertex \(u\in V\), define its neighbor sets as
$N(u)=\{w\in V : (uw)\in E\} \quad$ and $N'(u)=\{w\in V : (uw)\in E'\setminus E\}$. Then, for the single-layer Phantom-QAOA algorithm with parameters \(\beta\),\(\gamma\) and $\alpha$, the expected Max-Cut for each edge $(uv) \in E$ is given by
\begin{align}
\label{eqn:theorem}
    \langle C_{uv} \rangle &= \frac{1}{2} + \frac{\sin(4\beta)}{4} \sin (\gamma) \bigg((\cos^d (\gamma) \cos^{d'} (\alpha \gamma) \nonumber \\ & \quad\quad\quad\quad\quad+ \cos^e (\gamma) \cos^{e'} (\alpha \gamma) \bigg) \nonumber \\
    & - \frac{\sin^2(2\beta)}{4} \Bigg[\cos^{d+e-2f}(\gamma) \cos^{d'+e'}(\alpha\gamma) (1 - \cos^f (2\gamma)) \nonumber \\
    & + \cos^{d+e-f'}(\gamma) \cos^{d'+e'-f'}(\alpha\gamma) \frac{1}{2} \bigg( \nonumber \\ & \quad\quad \cos^{f'}(\gamma (1 - \alpha)) -\cos^{f'}(\gamma(1+\alpha)) \bigg) \nonumber \\
    & + \cos^{d+e}(\gamma) \cos^{d'+e'-2f''}(\alpha\gamma) (1 - \cos^{f''}(2\alpha\gamma))\Bigg].
\end{align}

where,
\begin{itemize}
    \item $d = |N(u) \setminus \{v\}|$, $d' = |N'(u)|$.
    \item $e = |N(v) \setminus \{u\}|$, $e' = |N'(v)|$
    \item $f$: the number of triangles containing $(uv)$ with edges $(uw)$, $(vw)$ in $G$.
    \item $f'$: the number of triangles containing $(uv)$ with one edge $(uw)$ in $G$ and one edge $(vw)$ in $G'/G$.
    \item $f''$: the number of triangles containing $(uv)$ with edges $(uw)$, $(vw)$ in $G'/G$.
\end{itemize}

\end{theorem}

\begin{proof}
Let $\tilde{C} = C + \alpha C'$, where $C$ is the original cost Hamiltonian for $G$ and $C'$ is the cost Hamiltonian for the edges in $G' \setminus G$. Note that $C_{uv} = \frac{1}{2}(I - Z_uZ_v)$ and therefore 
\begin{align}
    \langle \gamma, \beta | C_{uv} | \gamma, \beta \rangle = \frac{1}{2} - \frac{1}{2} \langle s | e^{i\gamma \tilde{C}} e^{i\beta B} Z_u Z_v e^{-i\beta B} e^{-i\gamma \tilde{C}} | s \rangle.
\end{align}
Using the Pauli-Solver algorithm detailed in \cite{hadfield2018quantumalgorithmsscientificcomputing}, we can express the expectation value $\langle Z_{u}Z_{v}  \rangle$ after one iteration of the mixer layer as:

\begin{align}
\label{eqn:expectation_formula}
    &\langle Z_{u}Z_{v}  \rangle \nonumber \\ &= \langle s | e^{i\gamma \tilde{C}} \left( c^2 Z_u Z_v + sc(Y_u Z_v + Z_u Y_v) + s^2 Y_u Y_v \right) e^{-i\gamma \tilde{C}} | s \rangle.
\end{align}

where the initial state $|s\rangle = |+\rangle^{\otimes n}$, and we let $c = \cos(2\beta)$ and $s = \sin(2\beta)$.

The term $Z_u Z_v$ commutes with the phase operator $e^{-i\gamma \tilde{C}}$ and, therefore, does not contribute to the expectation value. The remaining Pauli terms $Y_uZ_v, Z_uY_v$ and $Y_uY_v$ do not commute with the phase operator and thus contribute to the final expectation value.

The expectation value for the term $Y_u Z_v$ is given by:

\begin{align}
\label{eqn:YuZv_expectation}
    & \langle s | e^{-i\gamma \tilde{C}} (Y_u Z_v) e^{i\gamma \tilde{C}} | s \rangle \nonumber \\ &=
    \langle s | e^{-i\gamma Z_uZ_v} e^{-i\gamma \sum_{w} Z_u Z_w} e^{-i\gamma \sum_{w'} Z_v Z_{w'}} Y_u Z_v | s \rangle \nonumber \\
    &= \langle s | (\cos(\gamma)I - i\sin(\gamma)Z_u Z_v) \prod_{i=1}^{d} (\cos(\gamma)I - i\sin(\gamma)Z_u Z_{w_i}) \nonumber \\ & \quad \times \prod_{i=1}^{d'} (\cos(\alpha\gamma)I - i\sin(\alpha\gamma)Z_u Z_{w_i'}) | s \rangle,
\end{align}

where $d=|N(u) \setminus v|$, $w_i$ is a vertex in $N(u) \setminus v$, $d'=|N'(u)|$, and $w'_i$ is a vertex in $N'(u)$. The only term that contributes to the expectation value in Eq.~\ref{eqn:YuZv_expectation} is the one proportional to $Z_u Z_v \otimes I^{\otimes d} \otimes I^{\otimes d'} \cdot Y_u Z_v = -i X_u$. Thus, we have:

\begin{align}
\label{eqn:reduced_YuZv}
    \langle s | e^{i\gamma \tilde{C}} Y_u Z_v e^{-i\gamma \tilde{C}} | s \rangle &= \langle s | -i s' c'^d c''^{d'} (-i X_u) | s \rangle \nonumber \\
    &= -s' c'^d c''^{d'},
\end{align}

where we let $s' = \sin(\gamma)$, $c' = \cos(\gamma)$, and $c'' = \cos(\alpha\gamma)$.

Due to symmetry, the expression for the expectation value of the term $Z_uY_v$ will be the same as in Eq.~\ref{eqn:reduced_YuZv}, except now we use the variables $e = |N(v) \setminus u|$ and $e' = |N'(v)|$. Thus, 

\begin{align}
\label{eqn:triangle_free_modfied}
    & \langle s | e^{-i\gamma \tilde{C}} \left(Y_u Z_v + Z_u Y_v\right) e^{i\gamma \tilde{C}} | s \rangle \notag \\
     &= \sin \gamma \left(\cos^d \gamma \cos^{d'} \alpha \gamma + \cos^e \gamma \cos^{e'} \alpha \gamma \right).
\end{align}

For the final term in Eq. ~\ref{eqn:expectation_formula}, we can rewrite it as


\begin{align}
    &\langle s | e^{i\gamma \tilde{C}} Y_u Y_v e^{-i\gamma \tilde{C}} | s \rangle = \langle s | \prod_{i=1}^{d} (\cos(\gamma) I - i \sin(\gamma) Z_u Z_{w_i}) \nonumber \\
    &\quad \times \prod_{i=1}^{d'} (\cos(\alpha \gamma) I - i \sin(\alpha \gamma) Z_u Z_{w'_i}) \nonumber \\
    &\quad \times \prod_{j=1}^{e} (\cos(\gamma) I - i \sin(\gamma) Z_v Z_{w_j}) \nonumber \\
    &\quad \times \prod_{j=1}^{e'} (\cos(\alpha \gamma) I - i \sin(\alpha \gamma) Z_v Z_{w'_j}) Y_u Y_v | s \rangle.
\end{align}


The analysis is similar to the one detailed in \cite{hadfield2018quantumalgorithmsscientificcomputing} for deriving Eq. 5.14, but now we must consider the following three cases for a pair of vertices $(u, v) \in E$:

\begin{enumerate}
    \item Both edges $(u, w_i)$ and $(v, w_i)$ are contained in $G$.
    \item One of the edges, either $(u, w_i)$ or $(v, w_i)$, is contained in $G$, while the other edge is contained in $G'\setminus G$.
    \item Both edges $(u, w_i)$ and $(v, w_i)$ are contained in $G'\setminus G$.
\end{enumerate}


The lowest-order contributions from Case $1$ arise from terms proportional to 
\[
I^{\otimes (d+e-2)} \otimes I^{\otimes (d'+e')} \otimes (Z_u Z_w) \otimes (Z_v Z_w),
\]
with each such term corresponding to one triangle; there are \(f\) triangles of this type. Since for any given neighbor \(w_i\) we have
$Z_u Z_{w_i} \cdot Z_v Z_{w_i} = I$, any contribution involving an even number of these factors cancels out. Consequently, only terms with an odd number of sine factors (i.e., odd-order terms) contribute. The next nonzero contribution arises at third order, involving three pairs \((Z_u Z_{w_i}, Z_v Z_{w_i})\); there are \(\binom{f}{3}\) such terms, each proportional to \(s'^6\). This pattern continues for all higher odd orders, and we get 

\begin{align*}
    \binom{f}{1} c'^{d+e-2} c''^{d+e} s'^2 + \binom{f}{3} c'^{d+e-6}c''^{d+e}s'^6 + \\ \binom{f}{5} c'^{d+e-10}c''^{d+e}s'^{10} + \ldots
\end{align*}

Similarly, in Case 2, the first-order contributions are proportional to
\[
I^{\otimes (d+e-1)} \otimes I^{\otimes (d'+e'-1)} \otimes (Z_u Z_w) \otimes (Z_v Z_{w'}),
\]
corresponding to \(f'\) triangles. In Case 3, the first-order contributions are proportional to
\[
I^{\otimes (d+e)} \otimes I^{\otimes (d'+e'-2)} \otimes (Z_u Z_{w'}) \otimes (Z_v Z_{w'}),
\]
with \(f''\) triangles. As in Case 1, only odd-order terms contribute, and the higher-order contributions in both cases follow the same combinatorial pattern. We thus obtain the following:










\begin{align}
    &\langle s | e^{i\gamma C} Y_u Y_v e^{-i\gamma C} | s \rangle \nonumber \\ 
    &= c'^{d+e-2f} c''^{d'+e'} \sum_{i=1,3,5,\ldots} \binom{f}{i} (c'^2)^{f-i} (s'^2)^i \nonumber \\
    &+ c'^{d+e-f'} c''^{d'+e'-f'} \sum_{i=1,3,5,\ldots} \binom{f'}{i} (c'c'')^{f'-i} (s's'')^i \nonumber \\
    &+ c'^{d+e} c''^{d'+e'-2f''} \sum_{i=1,3,5,\ldots} \binom{f''}{i} (c''^2)^{f''-i} (s''^2)^i \nonumber
\end{align}

Applying the binomial theorem and simplifying via trigonometric identities leads to the following expression:


\begin{align}
\label{eqn:triangle_modified}
    &\langle s | e^{i\gamma C} Y_u Y_v e^{-i\gamma C} | s \rangle \nonumber \\
    &= c^{d+e-2f} c''^{d'+e'} (1 - \cos^f (2\gamma)) \notag \\
    &+ c^{d+e-f'} c''^{d'+e'-f'} \frac{1}{2} \left( \cos^{f'}(\gamma(1-\alpha)) - \cos^{f'}(\gamma(1 + \alpha) \right) \notag \\
    & + c^{d+e} c''^{d'+e'-2f''} (1 - \cos^{f''}(2\alpha\gamma)).
\end{align}

Substituting Eq. ~\ref{eqn:triangle_free_modfied} and Eq. ~\ref{eqn:triangle_modified} into Eq. ~\ref{eqn:expectation_formula} we obtain the final expression for $\langle C_{uv} \rangle$ in terms of $\gamma, \beta$ and $\alpha$.

\end{proof}

To facilitate later analysis, we present the triangle-free case (where both $G$ and $G'$ are triangle-free) as a corollary:

\begin{corollary}[Triangle-Free Case]
In the special case where the graph is triangle-free, meaning there are no triangles involving the vertices \(u\) and \(v\) such that \(f = f' = f'' = 0\), Eq.~\eqref{eqn:theorem} simplifies to:

\begin{align}
\label{eqn:corollary_triangle_free}
    &\langle C_{uv} \rangle \\
    &= \frac{1}{2} + \frac{\sin(4\beta)}{4} \sin \gamma \left(\cos^d \gamma \cos^{d'} (\alpha \gamma) + \cos^e \gamma \cos^{e'} (\alpha \gamma) \right).\notag
\end{align}

\end{corollary}

\subsection{Example}
\label{sec:analytic_example}

In this section, we analytically show an improvement in the approximation ratio for a cycle graph. This example demonstrates the utility of adding new edges in $G'$ such that additional triangles are formed with the existing edges in $G$.

\begin{figure}[t]
    \centering
    \begin{minipage}[t]{0.48\columnwidth}
        \centering
        \begin{tikzpicture}[scale=1]
            \foreach \i in {1,2,...,8} {
                \node[circle, draw, fill=black, inner sep=1pt, minimum size=0.5em] (N\i) at (\i*45:1.5) {};
            }
            \foreach \i in {1,2,...,8} {
                \pgfmathtruncatemacro{\j}{mod(\i,8) + 1}
                \draw[black, line width=0.7mm] (N\i) -- (N\j);  
            }

            \draw[red, line width=0.3mm] (N1) -- (N4);
            \draw[red, line width=0.3mm] (N4) -- (N7);
            \draw[red, line width=0.3mm] (N7) -- (N2);
            \draw[red, line width=0.3mm] (N2) -- (N5);
            \draw[red, line width=0.3mm] (N5) -- (N8);
            \draw[red, line width=0.3mm] (N8) -- (N3);
            \draw[red, line width=0.3mm] (N3) -- (N6);
            \draw[red, line width=0.3mm] (N6) -- (N1);
        \end{tikzpicture}
    \end{minipage}
    \hfill
    \begin{minipage}[t]{0.48\columnwidth}
        \centering
        \begin{tikzpicture}[scale=1]
            \foreach \i in {1,2,...,8} {
                \node[circle, draw, fill=black, inner sep=1pt, minimum size=0.5em] (N\i) at (\i*45:1.5) {};
            }
            \foreach \i in {1,2,...,8} {
                \pgfmathtruncatemacro{\j}{mod(\i,8) + 1}
                \draw[black, line width=0.7mm] (N\i) -- (N\j);
            }

            \draw[red, line width=0.3mm] (N1) -- (N3);
            \draw[red, line width=0.3mm] (N3) -- (N5);
            \draw[red, line width=0.3mm] (N5) -- (N7);
            \draw[red, line width=0.3mm] (N7) -- (N1);
            \draw[red, line width=0.3mm] (N2) -- (N4);
            \draw[red, line width=0.3mm] (N4) -- (N6);
            \draw[red, line width=0.3mm] (N6) -- (N8);
            \draw[red, line width=0.3mm] (N8) -- (N2);
        \end{tikzpicture}
    \end{minipage}
    \caption{\footnotesize Examples of 8-vertex cycle graphs $G$ shown in black and additional $\alpha$-weighted edges shown in red. The left image shows a $G'$ that is regular and triangle-free, while the right image shows a $G'$ that is also regular but contains triangles.}
    \label{fig:combined_graphs}
\end{figure}

Let $G$ be an arbitrary-sized cycle graph. Let's consider optimizing the Max-Cut expectation value for $G$ using the standard QAOA ansatz. For a $D$-regular triangle-free graph, the expectation value is given by

\begin{equation}
    \label{eqn:unweighted_regular}
    \langle C \rangle = \frac{m}{2} + \frac{m}{2} \sin 4\beta \sin \gamma \cos^{D-1} \gamma,
\end{equation}

where $m$ is the number of edges in the graph \cite{hadfield2018quantumalgorithmsscientificcomputing}. It is not difficult to show that the maximum expectation value occurs at $(\gamma^*, \beta^*) = \left(\arctan\left(\frac{1}{\sqrt{D-1}}\right), \frac{\pi}{8}\right)$, and thus, the maximum value of this expectation is

\[
    \max_{\gamma, \beta} \langle C \rangle = \frac{m}{2} + \frac{m}{2} \frac{1}{\sqrt{D}} \left( \frac{D-1}{D} \right)^{(D-1)/2}.
\]

Applying the above formula, the optimal parameters occur at $(\gamma^*,\beta^*) = (\pi/4, \pi/8)$ and the maximum cut value that can be achieved with the original ansatz is $\frac{3}{4}m$.

Now, let's consider the Phantom-QAOA ansatz with a phase operator reflecting a new graph $G'$. If we add the $\alpha$-weighted edges to $G$ such that $G'$ is triangle-free (left graph in Fig.~\ref{fig:combined_graphs}), then the QAOA cut expectation is symmetric across each edge in $G$. For the general case, where the cycle graph has any number of vertices and the $\alpha$-weighted edges connect all vertices that are three hops away, we apply Eq.~\ref{eqn:corollary_triangle_free} with parameters $d = e = 1$ and $d' = e' = 2$ for all edges. By summing the expectation $\langle C_{uv} \rangle$ over all edges in $G$, we obtain

\begin{align}
\label{eq:regular_triangle_free}
    \langle C \rangle = \frac{m}{2} + \frac{m}{2} \sin (4\beta) \sin (\gamma) \cos (\gamma) \cos^{2} (\alpha \gamma). \quad
\end{align}

Here, $m$ represents the total number of edges in the graph. For $\alpha = 0$, we recover the original formula in Eq.~\ref{eqn:unweighted_regular}. The optimal approximation ratio still occurs at $\alpha = 0$, and we see that the maximum cut remains $\frac{3}{4}m$.

Next, consider a $G'$ containing triangles (right graph in Fig.~\ref{fig:combined_graphs}). Again, due to symmetry, each edge has the same parameters $d, e$, and $f$. Substituting $d=e=1$, $d'=e'=2$, $f'=1$, and $f=f''=0$ into Eq. ~\ref{eqn:theorem} and summing over all edges for the general cycle graph gives

\begin{equation}
\label{eqn:8_node_cycle_add_hc_w_triangles}
\begin{split}
    \langle C \rangle &= \frac{m}{2} + \frac{m}{2} \sin\left(4\beta\right) \sin\left(\gamma\right) \cos\left(\gamma\right) \cos^{2}\left(\alpha \gamma\right) \\ &- m \sin^{2}(2\beta) \cos(\gamma) \cos^{3}\left(\alpha \gamma\right) \sin(\gamma) \sin(\alpha \gamma).
\end{split}
\end{equation}

Notice that the first two terms are the same for Eq.~\ref{eq:regular_triangle_free} and Eq.~\ref{eqn:8_node_cycle_add_hc_w_triangles},  but the third term allows for greater expressibility and leads to an improved approximation ratio for some $\alpha > 0$. To show this, we can resuse the optimal parameters $(\gamma^*, \beta^*) = (\pi/4, \pi/8)$ from the standard QAOA ansatz. The Max-Cut expectation as a function of $\alpha$ for the cycle graph becomes

\begin{align} 
    \langle C \rangle = \frac{m}{2} + \frac{m}{4} \cos^{2}\left(\frac{\pi}{4}\alpha\right) - \frac{m}{4}\cos^{3}\left(\frac{\pi}{4}\alpha\right) \sin\left(\frac{\pi}{4}\alpha\right).
\end{align}

It is easy to see that the function is periodic with a maximum cut value of $0.7925m$. For even-lengthed cycles, this results in an improved approximation ratio of $\frac{0.7925m}{m} = 0.7925$ for Phantom-QAOA, compared to the optimal approximation ratio of $\frac{0.75m}{m} = 0.75$ for the standard QAOA ansatz. Note that greater Max-Cut values can be found if we optimize $\gamma, \beta$ along with $\alpha$.

\section{Numerical Experiments}
\label{sec:numerical_experiments}

\subsection{Experimental Setup}

For our numerical experiments, we evaluated $\alpha$ values ranging from $0.0$ to $0.5$ in increments of $0.05$. For each $\alpha_i$, we optimized the angles $\gamma$ and $\beta$ within the ranges $\gamma \in (-\pi, \pi)$ and $\beta \in (-\pi/4, \pi/4)$ using the BFGS optimizer with $10$ random initial guesses. For $p = 1$, we used the optimal values $\gamma_i^*$ and $\beta_i^*$ for each $\alpha_i$ from this initial optimization.

For $p = 2$, we ran a second optimization pass where, for each $\alpha_i$, we used all the previously optimized parameters $\{\gamma_i^*, \beta_i^*\}_{i=1}^{k}$ from the first pass as starting points. This second pass allowed the optimizer to explore a wider range of starting points that were closer to the optimal parameters.

For the following results, we compare the two methods mentioned previously: the Full and Triangle methods. 
We applied both methods at depths $p=1$ and $p=2$ for nodes ranging from $6$ to $16$ and for all possible degrees\footnote{It is an elementary result of graph theory that it is impossible to have an $k$-regular graph with $n$ vertices if both $n$ and $k$ are odd.} for each node. The regular graphs were generated randomly using the \texttt{Graphs.SimpleGraphs.random\_regular\_graph} function in the Julia library \texttt{Graphs.jl} \cite{Graphs2021}. For each node and degree combination, we averaged the results over 10 non-isomorphic graphs whenever possible.

\begin{figure}[t]
  \centering
  \includegraphics[width=0.4\textwidth]{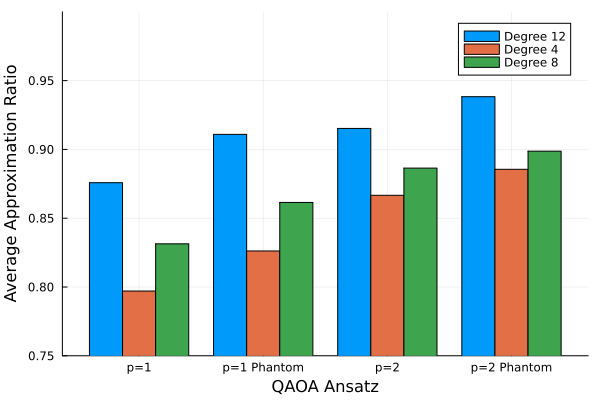}  
  \caption{\footnotesize The average approximation ratios on 16-node graphs with degrees 4, 8, and 12 across various QAOA ansatzes: the standard ansatz at p=1 and p=2 (denoted 'p=1' and 'p=2') and the Phantom-QAOA ansatz (denoted 'p=1 Phantom' and 'p=2 Phantom'). Here we used the Triangle method for Phantom-QAOA.}
  \label{fig:bar_plot}
\end{figure}

\subsection{Results}

Fig. ~\ref{fig:bar_plot} compares the average approximation ratios between the standard ansatz and our modified ansatz at depths $p=1$ and $p=2$. We observe that for various levels of sparseness, Phantom-QAOA offers a significant improvement in the approximation ratio over the standard ansatz. Additionally, we note that for higher degrees (degree 12, in this instance), the $p=1$ Phantom-QAOA ansatz performs as well as the standard $p=2$ QAOA. This similar trend was observed across various numbers of nodes. For lower degrees (i.e., when the graph is more sparse), the $p=1$ Phantom-QAOA ansatz yields approximation ratios that are approximately halfway between those of QAOA at $p=1$ and $p=2$.

As shown in the bar plot, the optimized approximation ratios generally cluster closely together, ranging from approximately 0.75 to 0.95. To facilitate easier interpretation, we plot the differences in approximation ratios. In Fig.~\ref{fig:avg_improvement_vs_degree} we show the average improvement in approximations ratio for various ansatzes for both the Full and Triangle methods at $p=1$ and $p=2$. We plot the average improvement as a function of graph degree $d$, with each line corresponding to a particular number of nodes $n$. As shown for both $p=1$ and $p=2$, the Triangle method consistently performs as well as or better than the Full method across all graph instances. For graphs with smaller degree, the Full method yields worse approximation ratios compared to the Triangle method, while for higher degrees, both methods yield similar approximation ratios.

When comparing the improvements for $p = 1$ and $p = 2$, we observe that the average improvement for $p = 1$ is approximately $0.04$, while for $p = 2$, it is around $0.02$. However, for $p = 2$, there is a slight upward trend in improvement as the degree increases, a pattern consistent across all values of $n$. It is also important to note that although the average improvement decreases from $p = 1$ to $p = 2$, the baseline approximation ratios at $p = 2$ are closer to $1$ (as seen by comparing the approximation ratios for the standard ansatz in ~\ref{fig:bar_plot}), and therefore there will be less room for improvement at $p = 2$. 

We also plot the value of $\alpha$ that yields the maximum improvement in the approximation ratio, denoted as $\alpha^*$ in Fig. ~\ref{fig:alpha_max_vs_degree}. The comparison between the Triangle and the Full methods follows a similar trend as for the approximation ratios: for higher degrees, $\alpha^*$ is the same for both methods. Additionally, we observe an upward trend in $\alpha^*$ as the degree increases, consistent across both methods and for both values of $p$.

\section{Discussion} 
\label{sec:discussion}

The original idea behind our proposed method was that the additional edges that are incorporated in the problem Hamiltonian would have weights that were small enough to allow the algorithm to ``see" more of the graph at a given depth $p$, while ensuring that the additional edges do not significantly impact the optimization of the target cost function. However, when the cost function remains the same, increasing $\alpha$ does not necessarily correspond to the increased likliehood of an edge being cut, but rather changes the periodicity of the cost landscape. This change in periodicity is touched upon in previous work related to QAOA for the weighted Max-cut problem \cite{Sureshbabu_2024, 10.1145/3584706, 9951313}.

To illustrate this effect on periodicity, in Fig.~\ref{fig:qaoa_search_space_comparison} we show a comparison of cost landscapes for the $12$-node cycle graph between the standard ansatz and Phantom-QAOA. We see that for the modified cost landscape that although the value of $\alpha=0.7$ is relatively high, there still exists peaks that yield an improved approximation ratio over the original ansatz as long as the parameter search spaces is expanded from $\gamma \in [-\pi, \pi]$, $\beta \in [-\pi/4, \pi/4]$ (green border) to $\gamma \in [-2\pi, 2\pi]$, $\beta \in [-\pi/2, \pi/2]$ (blue border). On the other hand, when the search is restricted to the original parameter space, the peak energy values for the Phantom-QAOA ansatz are lower than those for the standard ansatz. For practical purposes, we restrict our cost landscape to the region $\gamma \in [-\pi, \pi]$ and $\beta \in [-\pi/4, \pi/4]$, as optimizing over the entire possible region would be intractable. Also, the global optimum will on average occur near zero for non-periodic functions \cite{10.1145/3584706, Boulebnane2023, Brandhofer2022}. 

We also ran experiments on randomly generated Erdos-Renyi graphs with varying levels of probability. In most instances, we found that the approximation ratios did not increase when the optimization of $\gamma$ and $\beta$ was restricted to the region $\gamma \in [-\pi, \pi]$ and $\beta \in [-\pi/4, \pi/4]$. However, when the optimizer was allowed to explore the entire landscape, improvements similar to those observed with random regular graphs were observed. Additionally, it appeared that, for most graph instances, the approximation ratio would increase for almost any value of $0 < \alpha < 1$.

In our numerical experiments, we performed a fairly course-grained search over the $\alpha$ values, where for each $\alpha$ value, we optimized $(\gamma, \beta)$ using a classical optimizer. In practice, the parameter $\alpha$ could be optimized directly along with $\gamma, \beta$ using the classical optimizer. Considering $\alpha$ as a new parameter, the fact that we only introduce one parameter to the standard QAOA anstaz, regardless of $p$, is one advantage over other proposed ansatz variances, such as MA-QAOA and QAOA+ \cite{herrman2021multianglequantumapproximateoptimization, Chalupnik_QAOA+}.

One potential concern with our proposed QAOA ansatz is the number of two-qubit gates introduced into the circuit, as two-qubit gates are the primary source of error on modern NISQ devices. The number of two-qubit gates added is equal to the number of new edges in the intersection of $G'$ and the complement of $G$. Therefore, the number of additional two-qubit gates depends on the original graph and the method used (e.g., Full or Triangle), providing some flexibility in our approach. Users can adjust the ansatz based on their specific needs. Additionally, if we were to increase $p$ by 1 instead, the number of additional two-qubit gates would equal the number of edges in $G$.

\begin{figure}
    \centering
    \begin{subfigure}[b]{0.4\textwidth}
        \centering
        \includegraphics[width=\textwidth]{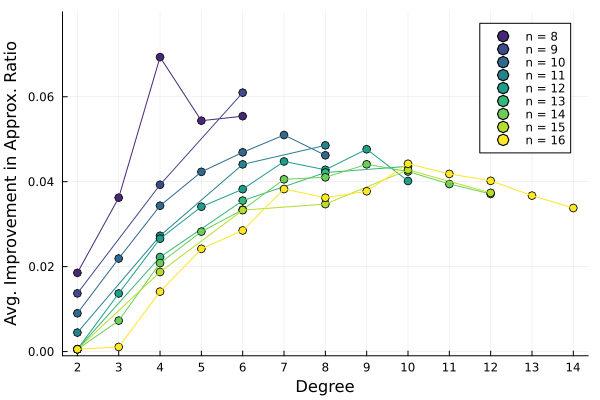}
        \caption{}
        \label{fig:avg_improvement_p1_full}
    \end{subfigure}
    \hfill
    \begin{subfigure}[b]{0.4\textwidth}
        \centering
        \includegraphics[width=\textwidth]{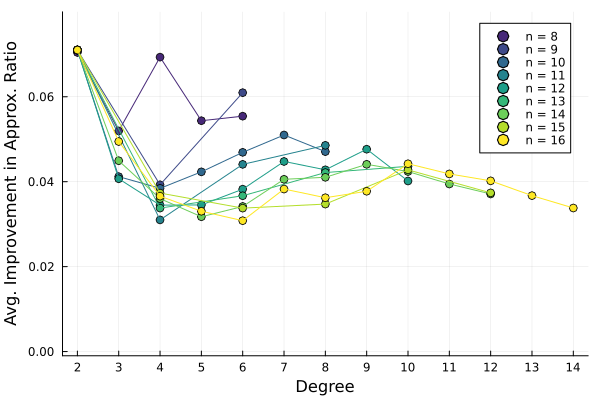}
        \caption{}
        \label{fig:avg_improvement_p1_triangle}
    \end{subfigure}
    \vfill
    \begin{subfigure}[b]{0.4\textwidth}
        \centering
        \includegraphics[width=\textwidth]{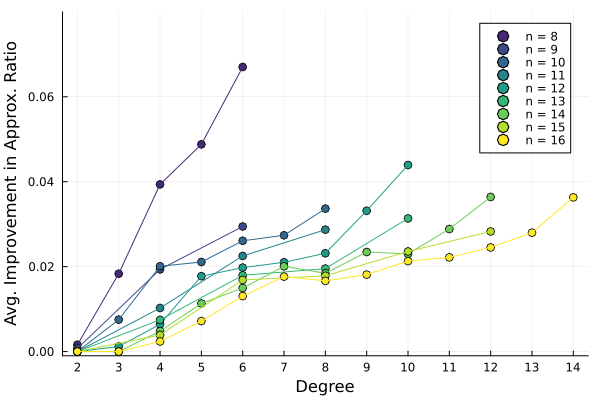}
        \caption{}
        \label{fig:avg_improvement_p2_full}
    \end{subfigure}
    \hfill
    \begin{subfigure}[b]{0.4\textwidth}
        \centering
        \includegraphics[width=\textwidth]{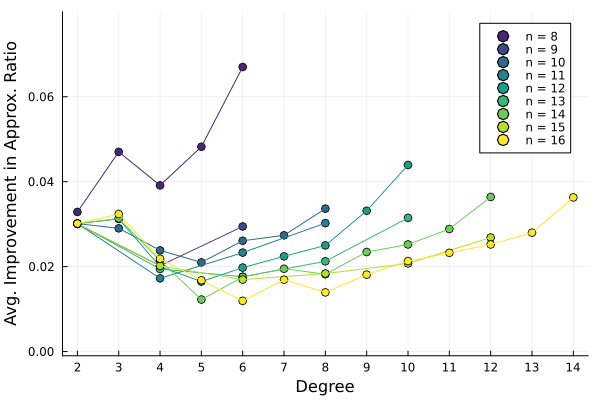}
        \caption{}
        \label{fig:avg_improvement_p2_triangle}
    \end{subfigure}
    \caption{\footnotesize Average improvement in approximation ratio vs degree for nodes 8-16: (a) p=1 Full method, (b) p=1 Triangle method, (c) p=2 Full method, (d) p=2 Triangle method}
    \label{fig:avg_improvement_vs_degree}
\end{figure}

\begin{figure}
    \centering
    \begin{subfigure}[b]{0.4\textwidth}
        \centering
        \includegraphics[width=\textwidth]{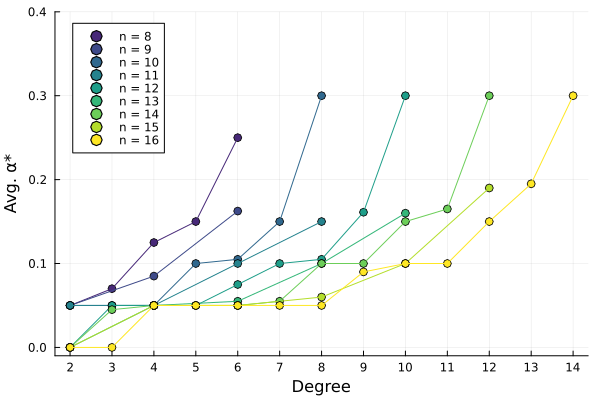}
        \caption{}
        \label{fig:alpha_max_p1_full}
    \end{subfigure}
    \hfill
    \begin{subfigure}[b]{0.4\textwidth}
        \centering
        \includegraphics[width=\textwidth]{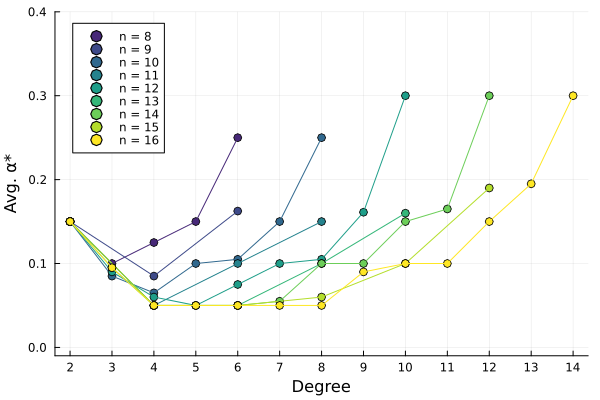}
        \caption{}
        \label{fig:alpha_max_p1_triangle}
    \end{subfigure}
    \vfill
    \begin{subfigure}[b]{0.4\textwidth}
        \centering
        \includegraphics[width=\textwidth]{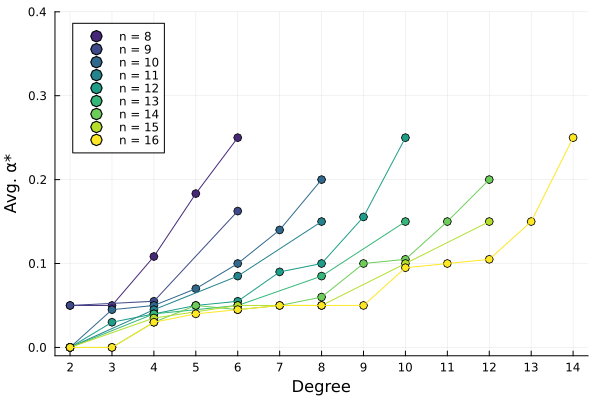}
        \caption{}
        \label{fig:alpha_max_p2_full}
    \end{subfigure}
    \hfill
    \begin{subfigure}[b]{0.4\textwidth}
        \centering
        \includegraphics[width=\textwidth]{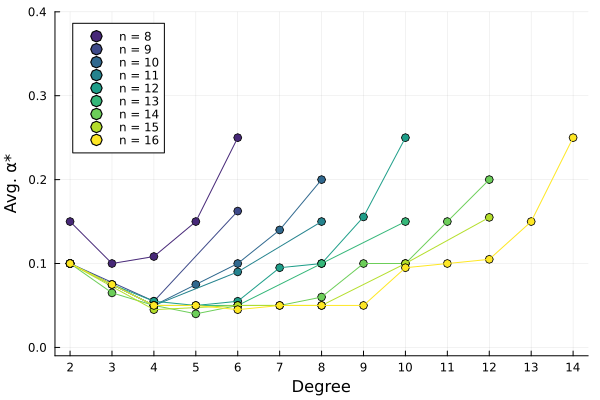}
        \caption{}
        \label{fig:alpha_max_p2_triangle}
    \end{subfigure}
    \caption{\footnotesize Average optimal $\alpha$ for (a) Full method at $p=1$, (b) Triangle method at $p=1$, (c) Full method at $p=2$, and (d) Triangle method at $p=2$.}
    \label{fig:alpha_max_vs_degree}
\end{figure}

Circuit depth is another important consideration. According to Vizing's theorem \cite{Vizing1965}, the minimum number of colors required to edge-color a graph---such that no two edges sharing the same vertex have the same color---is either $\Delta(G)$ or $\Delta(G)+1$ where $\Delta(G)$ is the maximum degree of the graph; moreover, an algorithm by Misra and Gries shows that an edge-coloring with $\Delta(G)+1$ colors can be found in polynomial time \cite{misra1992constructive}. In our context, for each QAOA layer, each edge in the graph will correspond to a 2-qubit gate. Then all qubit gates associated with edges of the same color can be executed simultaneously on quantum devices that allow for parallel circuit execution (e.g. some of IBM's superconducting quantum computers \cite{niu2022parallel}). Consequently, for each QAOA layer, the total circuit depth will be proportional to the number of colors required to edge-color the graph, which is bounded by the maximum degree of the graph. When going from the original graph $G$ to the modified graph $G'$ using the Triangle method, we have that $\Delta(G') \leq \Delta(G)^2$; note that this inequality is tight in the case where $G$ is a cycle of length at least 5. For graphs where this inequality is close to tight (e.g. random sparse graphs), this leads to a quadratic increase in the circuit depth which may possibly be too much in regards to dealing with noise. On the other hand, for increasingly denser graphs, we have that both $\Delta(G)$ and $\Delta(G')$ will approach $n-1$ implying that the increase in circuit depth will approach 0. 

From the graph-visibility perspective, intuition suggests that our approach would be ill-suited for dense graphs since QAOA would already be able to ``see" most of the graph at low circuit depths; however, the empirical results and the circuit depth analysis above indicate otherwise and in fact suggest that our approach is perhaps \emph{more} well-suited for dense graphs (compared to sparse graphs).


\begin{figure}[htb]
    \centering
    \begin{subfigure}[b]{0.4\textwidth}
        \centering
        \includegraphics[width=\textwidth]{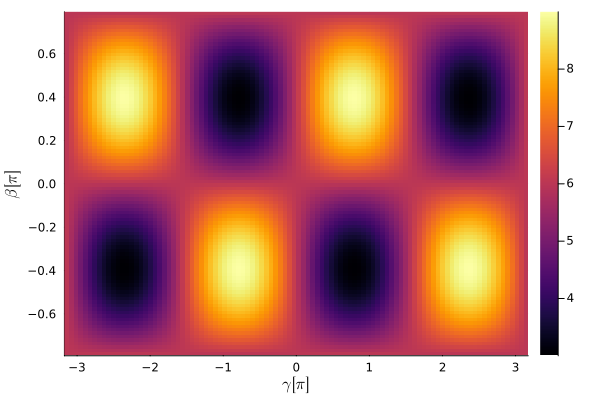}
        \caption{}
        \label{fig:standard_subqaoa}
    \end{subfigure}
    \hfill
    \begin{subfigure}[b]{0.4\textwidth}
        \centering
        \includegraphics[width=\textwidth]{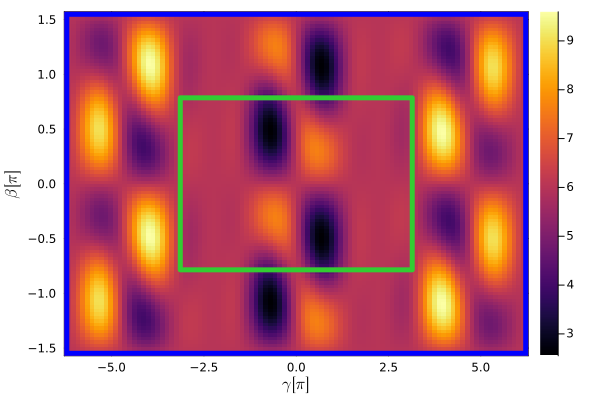}
        \caption{}
        \label{fig:modified_suboperator}
    \end{subfigure}
    \caption{\footnotesize Energy landscapes for (a) standard QAOA and (b) Phantom-QAOA with $G'$ constructed using the Triangle method, setting $\alpha=0.7$. The region within the green border is the region we explore in our numerical experiments. We see that greater peak energies can be found outside of the original search space.}
    \label{fig:qaoa_search_space_comparison}
\end{figure}

\section{Conclusions and Future Directions}
\label{sec:conclusions}

In this paper, we present a new QAOA ansatz that incorporates one additional parameter, regardless of system size. We derive a general formula for the expectation value of this new ansatz for $p=1$, and numerically show an improvement in the approximation ratio over the standard QAOA ansatz for $p=1$ and $p=2$. Additionally, both our analytical and numerical results in this paper suggest that the additional weighted edges should be connected to nodes in $G$ as to form triangles, as additional edges elsewhere in the graph do not appear to yield an improvement to the optimal approximation ratio, at least for $p=1$ and $p=2$.

One potential modification to the ansatz is to increase the number of parameters by assigning different weights to the additional edges rather than using a single weight for all of them. This approach would allow for greater flexibility and expressibility in the ansatz, potentially leading to higher maximum approximation ratios than those achievable with a single weight for all the additional edges. However, this increased expressibility would come at the cost of a more challenging optimization. Future work could explore balancing the number of parameters and the complexity of optimization to achieve optimal performance.

Another potential direction for future research is to explore alternative modifications to $G'$, beyond the Full and Triangle methods discussed in this paper. It is possible that for larger circuit depths, the optimal placement of the additional $\alpha$-weighted edges in $G'$ may differ from the configurations that are most effective at $p=1$ and $p=2$ and for the graph sizes considered in this work. As the QAOA algorithm's depth increases, its ability to "see" the graph also expands. This expanded visibility could mean that connecting different nodes, or even using more complex patterns of additional edges, might lead to better approximation ratios for deeper circuits.

Lastly, to improve the practical utility of our new ansatz, a potential direction for future work could involve optimizing the trade-off between the number of edges added to the graph and the resulting improvement in the approximation ratio. One of the primary challenges in implementing our approach on real quantum hardware is the increase in the number of two-qubit gates required in the quantum circuit when additional edges are introduced, as adding too many can significantly impact the performance and scalability of the algorithm. Therefore, it would be valuable to further investigate how many additional edges should be introduced and where to place them strategically to maximize the improvement in the approximation ratio while minimizing the overhead in circuit complexity. This could involve developing intelligent algorithms to select the most impactful edges to add, thereby optimizing both the quantum resources and the algorithm's performance.

\section*{Acknowledgment}
This work was supported by the U.S. Department of Energy through the Los Alamos National Laboratory. Los Alamos National Laboratory is operated by Triad National Security, LLC, for the National Nuclear Security Administration of U.S. Department of Energy (Contract No. 89233218CNA000001). The research presented in this article was supported by the Laboratory Directed Research and Development program of Los Alamos National Laboratory under project number 20230049DR as well as by the NNSA's Advanced Simulation and Computing Beyond Moore's Law Program at Los Alamos National Laboratory. Report Number: LA-UR-24-30655.

\bibliographystyle{plain}
\bibliography{references}

\end{document}